\documentclass[conference]{IEEEtran}
\usepackage{graphicx,amssymb,amsmath}

\newtheorem{observation}{Observation}
\newtheorem{lemma}{Lemma}
\newtheorem{theorem}{Theorem}

\begin{document}
%
\title{Sixsoid: A new paradigm for $k$-coverage in 3D Wireless Sensor Networks}

\author{\IEEEauthorblockN{Nabajyoti Medhi}
\IEEEauthorblockA{Department of Computer Science and Engineering\\
National Institute of Technology, Meghalaya\\
Shillong, INDIA\\
Email: nabajyotimedhi21@gmail.com}
\and
\IEEEauthorblockN{Manjish Pal}
\IEEEauthorblockA{Department of Computer Science and Engineering\\
National Institute of Technology, Meghalaya\\
Shillong, INDIA\\
Email: manjishster@gmail.com}
}

%


\maketitle

\begin{abstract}
Coverage in 3D wireless sensor network (WSN) is always a very critical issue to deal with. Coming up with good coverage models implies energy efficient networks. $K$-coverage is a model that ensures that every point in a given 3D Field of Interest (FoI) is guaranteed to be covered by $k$ sensors. In the case of 3D, coming up with a deployment of sensors that gurantees $k$-coverage becomes much more complicated than in 2D. The basic idea is to come up with a convex body that is guaranteed to be $k$-covered by taking a specific arrangement of sensors, and then fill the FoI with non-overlapping copies of this body.  
In this work, we propose a new geometry for the 3D scenario which we call a \textbf{Sixsoid}. Prior to this work, the convex body  which
was proposed for coverage in 3D was the so called \textbf{Reuleaux Tetrahedron} \cite{A, AD1}. Our construction is motivated from a construction that can be applied to the 2D version of  the problem \cite{MS} in which it implies better guarantees over the \textbf{Reuleaux Triangle}.  Our contribution in this paper is twofold, firstly we show how Sixsoid gurantees more coverage volume over Reuleaux Tetrahedron, secondly we show how Sixsoid also guarantees a simpler and more pragmatic deployment strategy for 3D wireless sensor networks. In this paper, we show the construction of Sixsoid, calculate its volume and discuss its implications on the $k$-coverage in 3D WSNs.
\end{abstract}


\begin{IEEEkeywords}
Wireless Sensor Networks, $k$-coverage, Reuleaux Tetrahedron, Sixsoid.
\end{IEEEkeywords}

%
\IEEEpeerreviewmaketitle

\section{Introduction}
Wireless Sensor Networks (WSN) consists of energy constrained sensor nodes having limited sensing range, communication range, processing power and battery. The sensors generally follow different hop-by-hop ad-hoc data gathering protocols to gather data and communicate. The sensors can sense the information whichever lies in its sensing range using RF-ID (Radio Frequency Identification). The sensed data can be communicated to another sensor node which lies within communication range of the sender. 
The gathered data finally reaches the base station which may be hops apart from any sensor. Since sensor transceivers are omni-directional, we assume the sensing and communication ranges as spheres of certain radii. The network is called homogeneous when all the sensors have the same radii and heterogeneous otherwise. Coverage of a certain FoI and deployment of sensors are an issue of research where the aim is to make energy efficient networks. Sensor deployment and coverage in 2D requires simpler strategies and protocols as compared to 3D. 3D sensor network is used generally for underwater sensor surveillance, floating lightweight sensors in air and space, air and water pollution monitoring, forest monitoring, any other possible 3D deployments etc. Real life applications of WSNs are mostly confined to 3D environments. The term $k$-coverage in 3D is used to describe a scenario in which the sensors are deployed in such a way that $k$ sensors cover a common region. More precisely, a point in 3D is said to be $k$-covered if it lies in the region that is common to the sensing spheres of $k$-sensors, $k$ being termed as degree of coverage. Indeed, the ultimate aim of this project is to \textbf{come up with a deployment strategy for sensors that guarantees $k$-coverage of a given 3-dimensional Field of Interest (FoI) for large values of $k$}. A first step to address this issue is to come up with a 3-dimensional convex body (tile) that is guaranteed to be $k$-covered by a certain arrangement of $k$ (or more) sensors, and then fill the FoI with non overlapping copies of that shape by repeating the same arrangement.

The term Sixsoid has been coined in this paper to signify a geometrical shape that resembles a super-ellipsoid \cite{W}. 
Sixsoid is created by the intersection of six sensors, each having the same sensing radius, which are placed on the six face centers of a cube of side length $r$ where $r$ is the radius of the sensing spheres. We compare the implications of this convex body with the previously proposed model on 3D $k$-coverage based on Reuleaux Tetrahedron. Recall that the Reuleaux Tetrahedron, is created by the intersection of four spheres placed on the vertices of a regular tetrahedron of side length $r$. In an attempt to guarantee 4-coverage of the given field, \cite{AD1} considers a scenario in which four sensors are placed on the vertices of a regular tetrahedron of side length equal to the sensing radius $r$. It is well known that the volume of 
of a Reuleaux Tetrahedron constructed out of a regular tetrahedron of side length $r$ is approximately $0.422r^3$ \cite{RT}. Unfortunately it is not
possible to obtain a tiling of the 3D space with non-overlapping copies of Reuleaux Tetrahedron \cite{A}. In fact, such a tiling is not
possible even with a tetrahedron \cite{MW}. In \cite{A} a plausible deployment strategy is hinted that exploits this construction by overlapping two Reuleaux Tetrahedrons, gluing them at a common tetrahedron's face, but this deployment doesn't seem to be pragmatic.  In this paper, we propose another 3D solid (the Sixsoid) for this purpose and an extremely pragmatic deployment strategy. We show that using our deployment strategy one is guaranteed to have 6-coverage of approximately 68.5\% and 4-coverage of 100\%  of a given 3D polycubical Field of Interest which is a significant improvement over the gurantees provided in \cite{A, AD1}.

\subsection{Prior Work on 3D $k$-coverage}
There are relatively fewer works done to address the problem of 3D $k$-coverage as compared to the 2D version of the problem. In \cite{A,AD1} the authors make significant progress on this problem. They discuss the relevance of Reuleaux Tetrahedron in various issues dealing with connectivity and $k$-coverage in 3D. Prior to that the following works disscuss the 3D version; in \cite{XLL}, authors propose a coverage optimization algorithm based on sampling for 3D underwater WSNs. \cite{MKPS} proposes an optimal polynomial time algorithm based on voronoi diagram and graph search algorithms, in \cite{HT} authors suggest algorithms to ensure $k$-coverage of every point in a field  where sensors may have same or different sensing radii, \cite{KLB} defines the minimum number of sensors for $k$-coverage with a probability value, \cite{WY} studies the effect of sensing radius on the probability of $k$-coverage, \cite{BKXYL} proposes an optimal deployment strategy to deal with full coverage and 2-connectivity, \cite{AH} brought forward a sensor placement model based on Voronoi structure to cover a 3D region, in \cite{R}, authors provide a study on connectivity and coverage issues in a randomly deployed 3D WSN. In this paper, we compare our proposed Sixsoid based $k$-coverage model with the existing Reuleaux Tetrahedron based model \cite{A}. Comparison has been done in terms of volume of $k$-coverage, sensor spatial density in 3D and placement and packing in 3D. 

\section{Our Contribution}\label{oc}
In a previous work \cite{MS}, the authors used a convex body that resembles a super-ellipse for $k$-coverage in 2D homogeneous WSN. This model \cite{MS} showed much better efficiency in terms of area of $k$-coverage, energy consumption and requirement of less number of sensors as compared to the $k$-coverage model with Reuleaux Triangle \cite{A}. This work done in \cite{A} is further extended to 3D by taking Reuleaux Tetrahedron. Motivated by the fact that considering a ``super-elliptical'' tile has proven to be much better than the Reuleaux Triangle based model \cite{MS}, in this work we extend that idea to 3D. The main hurdle was to compute the volume of the solid 
(Sixsoid) generated out of our construction. Unlike the Reuleaux Tetrahedron its volume was not already known. Another reason for opting our construction is its resemblence to a superellipsoid because of which it fits much better in 3D rather than a Reuleaux Tetrahedron. Moreover, Sixsoid provides a practical and easier packing in 3D. Once we fill the FoI with cubes of side length $R$, a Sixsoid is formed by the overlapping of six sensors lying on the six face centers of a cube. So, every cube contains a Sixsoid inside it. When we consider packing in 3D, cubes are space filling. We propose a result in a later section that states that our deployment strategy ensures 4-coverage of the entire FoI along with fact that approximately 68\% of it is 6-covered. Thus, the Sixsoid based model ensures at least 4-coverage of the entire FoI. Packing of Reuleaux Tetrahedron in 3D as proposed in \cite{A} is harder to achieve and it is not feasible for practical deployment. 
In the subsequent sections, we calculate the volume of the Sixsoid with the supporting theorems and compare our results with the Reuleaux Tetrahedron based model. In the remaining part of the paper, we discuss the computation of volume of a Sixsoid in Section 3. In Section 4 we discuss the deployment and packing strategy in 3D that exploits the structure of Sixsoid. Section 5 is devoted to the comparison of our model with the results proposed in \cite{AD1}. We end the paper with conclusion and potential research directions for future. 

\section{Computing the Volume of Sixsoid}
The \textbf{Sixsoid} which we denote by $\mathcal{S}(R)$ is defined as the 3-D object obtained due to the intersection of
six spheres placed on the centers of the faces of a cube of side-length $R$. The radius of
each sphere is $R$. In this section, we compute the volume $V$ of $\mathcal{S}(R)$.
Our basic strategy is to take a horizontal cross section (via a sliding plane) of the cube at a distance of $x$ from the top of
the cube and perform the following integration,
\begin{eqnarray*}
V = \int^{R}_{0} A(x)\,dx = 2 \int^{R/2}_{0} A(x)\,dx
\end{eqnarray*}
where $A(x)$ is the area of the cross section and the last equality follows due to the 
symmetry of $\mathcal{D}(R)$.

Henceforth, we will use the following notation regarding the geometry of
the objects involved in this section. Let $F_{x_1}$(top), $F_{x_2}$(bottom), $F_{y_1}$(left), $F_{y_2}$(right),
$F_{z_1}$(front), $F_{z_2}$(back) be the faces of the cube and  we denote the respective spheres (centered at face centers) 
by $S_{x_1},S_{x_2}, S_{y_1},S_{y_2}, S_{z_1},S_{z_2}$. We will denote the sliding plane by $P_{x}$ and let $r_{\alpha_i}$ be its radius. Let $C_{\alpha_i}$ be the circle which is obtained by the intersection of $P_x$ with $S_{\alpha_i}$ where $\alpha\in \{x,y,z\}$ and $i \in \{1,2\}$. In the remainder of this section we show the behavior of the cross-section as the plane $P_x$ slides from $x = 0$ to $x = \frac{R}{2}$.

\begin{observation}
 For all $0 \leq x \leq R/2$, $C_{x_1}$ and $C_{x_2}$ are concentric circles and $r_{x_2} \leq r_{x_1}$.
Also as $x$ varies from $0$ to $R/2$, $r_{x_1}$ monotonically decreases, $r_{x_2}$ monotonically increases and $C_{x_1} = C_{x_2}$ at
$x = R/2$.
\end{observation}

\begin{observation}
For $0 \leq x \leq R/2$, $r_{y_1} = r_{y_2} = r_{z_1} = r_{z_2} = \sqrt{\frac{3R^2}{4} + Rx - x^2}$.
\end{observation}

\begin{lemma}
$A(x) = \pi(2Rx - x^2)$ from $x = 0$ to $x = l_1$, where $l_1 = R(\frac{3 - \sqrt{7}}{4}) \approx 0.0886R$
\end{lemma}
\begin{proof}
We observe that, $A(x)$ is exactly $C_{x_2}$ as $P_x$ slides below $x = 0$. So to compute $A(x)$ we only need to compute 
the radius $r_{x_2}$ of $C_{x_2}$. This value can be computed as
$\sqrt{R^2 - (R-x)^2} $, which proves the claim regarding $A(x)$. Now inorder to know $l_1$ (the event
when the cross section changes), we need to find the value of $x$ at which $C_{x_2}$ is tangential to $C_{z_1}$, this is also
the instance at which $C_{x_2}$ is tangential to each $ C_{y_1},C_{y_2}$ and $C_{z_1}$. To find this $x$ we need to solve the
following equation, $r_{x_2} = R/2 + r_{z_1}$
\small
\begin{eqnarray*}
\sqrt{\frac{3R^2}{4} + Rx -x^2} &=& \sqrt{2Rx - x^2} + \frac{R}{2} \\
\frac{3R^2}{4} + Rx -x^2 &=& 2Rx-x^2 +\frac{ R^2}{4} + R\sqrt{2Rx - x^2} \\
\frac{R}{2} - x &=&  \sqrt{2Rx -x^2}\\
2x^2 - 3Rx + \frac{R^2}{4} &=& 0
\end{eqnarray*}
The root of the above quadratic equation which is less than $R/2$ is $x = \frac{(3 - \sqrt{7})R}{4}$. This proves the lemma.
\end{proof}

\begin{lemma}
$A(x) =  \frac{1}{2} \cdot r^2 \cdot \left( \frac{\pi}{2} - 2 \sin^{-1} \left( \frac{Y}{2\sqrt{2Rx - x^2}}\right) \right)$  
+ $\frac{1}{2}\cdot Y \cdot \sqrt{r^2 - \frac{Y^2}{4}} + \left( \frac{1}{2}\cdot r^2 \cdot 2 \sin^{-1}\left( \frac{Y}{2r_{y_1}}\right) \right)$ 
$- \left(\frac{1}{2}\cdot Y \cdot  \sqrt{{r_{y_1}}^2 - \frac{Y^2}{4}}\right)$ where $Y = 2\sqrt{\frac{7R^2}{4} + 3Rx -2x^2},
r = \sqrt{2Rx-x^2}$ and $r_{y_1} = \sqrt{Rx-x^2+\frac{3R^2}{4}}$ from $x = l_1$ to $x = l_2$ where 
$l_1 = R(\frac{3 - \sqrt{7}}{4})$ and $l_2 =  (\frac{2}{3} - \frac{\sqrt{10}}{6})R$
\end{lemma}

\begin{proof}
Recall that above $x = l_1$, the cross-section is a circle and at $x = l_2$,  $C_{x_2}$ is tangential to $C_{z_1}$. As $P(x)$
slides slightly below $x = l_1$, the cross-section is a 8-sided closed region drawn in Fig.\ref{z2}. Let $O$
be the center of the square and $\angle EOF =\angle GOH = \angle IOJ = \angle KOL = \alpha$, $\angle FOG = \angle HOI = \angle
JOK = \angle LOE$.
\begin{figure}
\centering
\includegraphics[scale = 0.25]{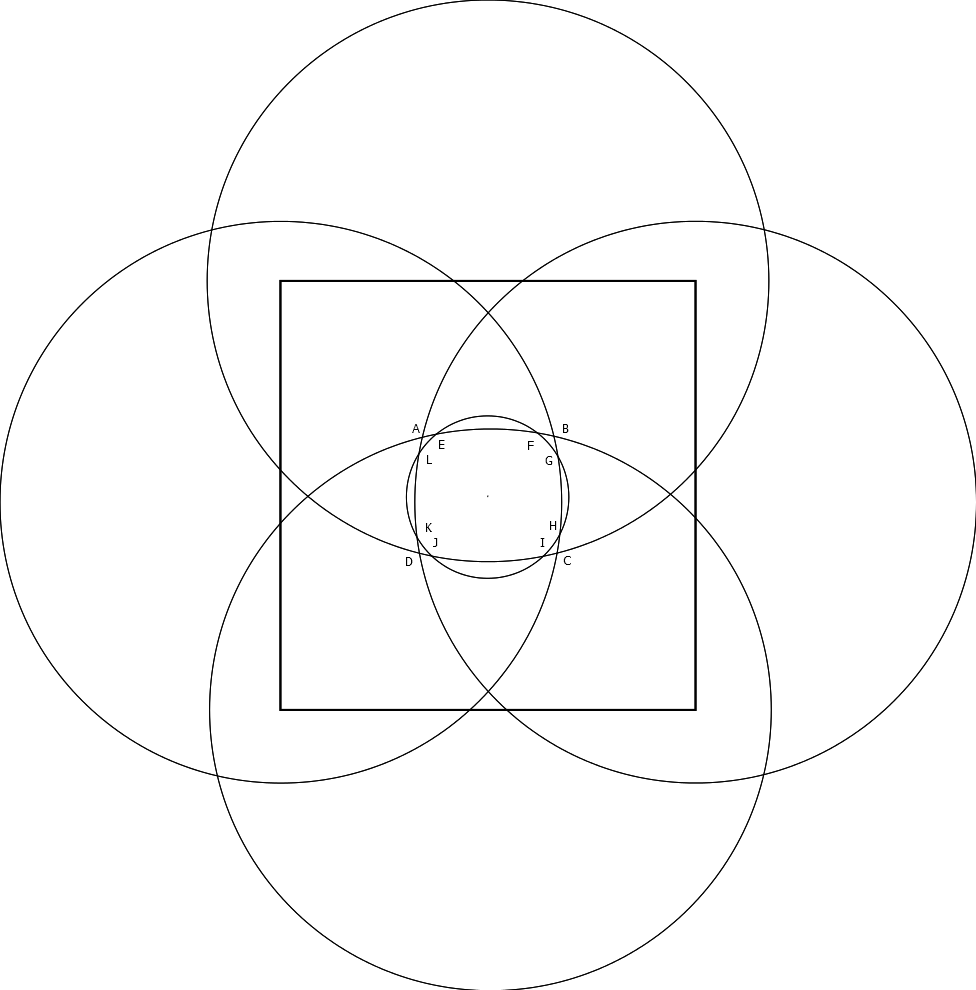}
\caption{Cross-Section of $A(x)$ is the enclosed region $EFGHIJKLE$, for $l_1 \leq x \leq l_2$}
\label{z2}
\end{figure}
Our aim is to find out the area of  the region $EFGHIJKLE$. From the notation used in the Fig. \ref{z2} we deduce that
$\alpha + \theta = \frac{\pi}{2}$. Also it is easy to see that $\alpha = 2 \sin^{-1} \left(\frac{GH}{2r}\right)$ where $r = r_{x_2}$ and
$GH$ is the length of the line segment joining $G$ and $H$. We first show how to compute $GH$ in terms of $R$ and $x$. \\
\textbf{Computing $GH$}:
Notice that $G$ and $H$ are points of intersection of the circles $C_{y_1}$ and $C_{x_2}$. Let us write the equation
of these circles assuming $O$ to be the origin and lines parallel to $y$ and $z$ axis passing though it as $x'$ and $y'$ axes.
The equations of $C_{y_1}$ and $C_{x_2}$ in this coordinate system would be $x'^{2} + (y'-\frac{R}{2})^{2} = r_{y_1}^2$ and
$ (x'-\frac{R}{2})^{2} +  (y'-\frac{R}{2})^{2} = r_{x_2}^2$ respectively. Subtracting the later from the former gives
$x' = \frac{1}{R}\left(r_{y_1}^2 - r_{x_2}^2 + \frac{R^2}{4} \right)$. Now if we conside the equation of $C_{y_1}$
as a quadratic equation in $y'$ with $x' = \frac{1}{R}\left(r_{y_1}^2 - r_{x_2}^2 + \frac{R^2}{4} \right)$ then the modulus of the difference of the
roots of this equation will exactly be the length of $GH$. Simple algebraic manipulation shows that
\begin{eqnarray*}
GH = 2\sqrt{3Rx - R^2/4 - 2x^2}
\end{eqnarray*}

The desired area can be decomposed into two parts. (i) 4 times the area
of sector $EOG$ (ii) 4 times the area of region $GOH$. We compute each of these as follows:
\begin{itemize}
\item[(a)] Area of sector $EOG$ is
\small
\begin{eqnarray*}
r^2 \cdot \theta/2 &=& \frac{1}{2} \cdot \left(2Rx - x^2 \right) \cdot \left( \frac{\pi}{2} - 2 \sin^{-1} \left( \frac{GH}{2\sqrt{2Rx - x^2}}\right) \right)
\end{eqnarray*}
\item[(b)] Area of the region $GOH$ =  Area of ($\Delta GOH$) + Area of the Cap, which is equal to \\
$\frac{1}{2}\cdot GH \cdot \sqrt{r^2 - \frac{GH^2}{4}} + \left( \frac{1}{2}\cdot r^2 \cdot 2 \sin^{-1}\left( \frac{GH}{2r_{y_1}}\right) \right)$ \\
$- \left(\frac{1}{2}\cdot GH \cdot  \sqrt{{r_{y_1}}^2 - \frac{GH^2}{4}}\right)$

\end{itemize}
Adding the aforementioned two values, multiplying by 4, and replacing the values of $GH, r, r_{y_1}$, gives us the expression of $A(x)$ in the lemma. \\
\textbf{Computing $l_2$}:
To find the value of $x$ at which the cross-section changes from the above mentioned 8-sided region, we need
to look at the instant when the circle $C_{x_2}$ circumscribes the region formed due the intersection of
$C_{y_1}, C_{y_2}, C_{z_1}, C_{z_2} $. At this instant the following equality holds;
\small
\begin{eqnarray*}
r_{x_2} &=& \sqrt{2} \left( z - R/2 \right) \\
&=& \sqrt{2}\left( \frac{R + \sqrt{5R^2 + 8Rx - 8x^2}}{4} - \frac{R}{2}\right) \\
& =&  \sqrt{2}\left( \frac{\sqrt{5R^2 + 8Rx - 8x^2}}{4} - \frac{R}{4}\right) \\
\sqrt{2Rx -x^2} &=& \sqrt{2}\left( \frac{\sqrt{5R^2 + 8Rx - 8x^2}}{4} - \frac{R}{4}\right) \\
5R^2 - 8x^2 + 8Rx &=& \sqrt{16Rx - 8x^2} + R \mbox{ (rearranging terms) } \\
6x^2 - 8Rx + R^2 &=& 0   \mbox{ (after a few algebraic manipulations) }
\end{eqnarray*}
which solves to $x = (\frac{2}{3} + \frac{\sqrt{10}}{6})R$ and $x = (\frac{2}{3} - \frac{\sqrt{10}}{6})R$, since $0 \leq x \leq R/2$, the only root which is of concern for
us is $x =  (\frac{2}{3} - \frac{\sqrt{10}}{6})R$.
\end{proof}

\begin{lemma}
$A(x) = 4\left[\frac{1}{2}\cdot s^2\cdot \theta - \frac{1}{2} \cdot z \cdot (2z - R)\right] + 
4\left[(y - \frac{R}{2})\right]^2$.
where $\theta = \angle AOB = 2\sin^{-1}\frac{z-\frac{R}{2}}{z}$ and $s = \sqrt{Rx - x^2 + \frac{3R^2}{4}}$, $z =  \frac{R + \sqrt{5R^2 - 8x^2 + 8Rx}}{4}$.
\end{lemma}
\begin{proof}
We notice that for $l_2 \leq x R/2$, $A(x)$ is the area which is common to $C_{y_1} , C_{y_2},  C_{z_1}, C_{z_2}$. Fig.\ref{z7} shows the geometry of $A(x)$. Again let $O$ be the center of the square
\begin{figure}
\centering
\includegraphics[scale = 0.36]{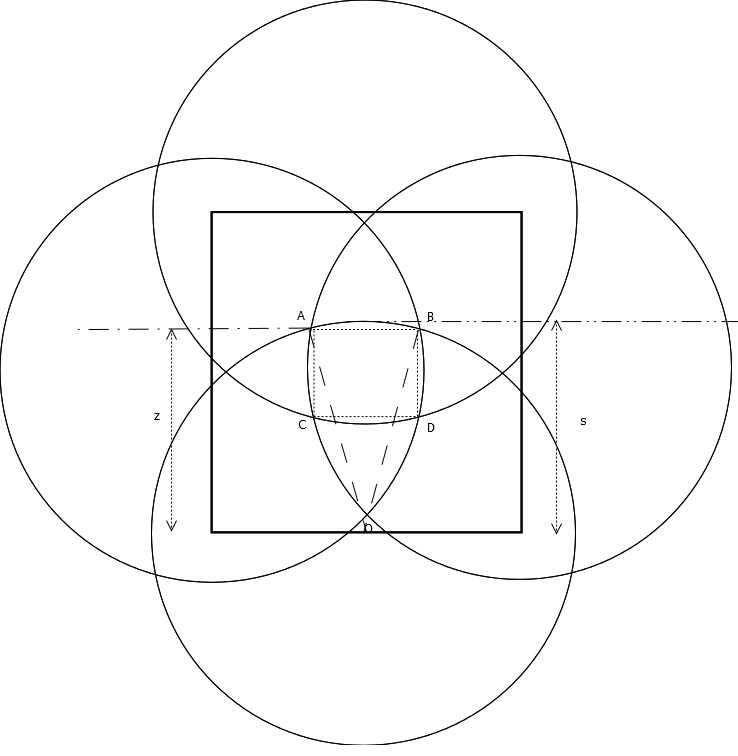}
\caption{Cross-Section of $A(x)$ is the enclosed region $ABCD$ for $l_2 \leq x \leq R/2$}
\label{z7}
\end{figure}
We divide it into two parts, (i) The square $ABCD$ (ii) Four caps surrounding the square each of equal area. We compute
each of these areas as follows:
\begin{itemize}
\item[(a)] The area of the square is $l^2$ where $l$ is the side length of the square, which according to the figure is $\{2\left(z-\frac{R}{2}\right)\}^2$. It
can be verified that $z = \frac{R + \sqrt{5R^2 + 8Rx - 8x^2}}{4}$ and $s = r_{z_1} = \sqrt{\frac{3R^2}{4} + Rx -x^2}$. Thus area of the square is
\begin{eqnarray*}
\left(2 \cdot \frac{\sqrt{5R^2 - 8x^2 + 8Rx}- R}{4}\right)^2
\end{eqnarray*}

\item[(b)] The area of a cap (say the cap $AB$) can be obtained by subtracting the area of $\Delta OAB$ from the area of sector $OAB$
which is equal to
\begin{eqnarray*}
\frac{1}{2}\cdot s^2\cdot \theta - \frac{1}{2} \cdot z \cdot (z - \frac{R}{2})
\end{eqnarray*}
where $\theta = \angle AOB = 2\sin^{-1}\frac{z-\frac{R}{2}}{z}$ and \\
 $s = \sqrt{Rx - x^2 + \frac{3R^2}{4}}$, $z =  \frac{R + \sqrt{5R^2 - 8x^2 + 8Rx}}{4} $ 

\end{itemize}
Putting the values of $z$ and $s$ and multiplying the area of the cap by $4$ gives us the desired result.
We use observation 1, to conclude that the cross section doesn't change for all values  of $x$ between $l_3$
and $R/2$. Thus $l_4 = R/2$.
\end{proof}

As a consequence of the previous results, we can conclude the following:

\begin{theorem}

\begin{eqnarray*}
 V = 2 \left[ \int^{l_1}_{0} A(x)\,dx +  \int^{l_2}_{l_1} A(x)\,dx +  \int^{R/2}_{l_2} A(x)\,dx \right]
\end{eqnarray*}
\end{theorem}

In order to get a good estimate on the value of the above definite integral, we use Matlab to perform the 
numerical integration. The value of $V$ turns out to be approximately $0.685R^3$.

\subsection{Computing the 4-covered Volume}
To compute the volume of the FoI we need to compute the volume 
that is only 3-covered, as the remaining volume is 4-covered. We notice that as the
horizonal plane slides from $x =0$ to $x = R/2$, the topology of the
cross section that is only 3-covered changes twice, thus to compute the volume
inside the cube that is only 3-covered have to perform an integration similar to 
case of the Sixsoid. It can be verified the two transitions happen at $x = (\frac{2}{3} - \frac{\sqrt{10}}{6})R$
and $x = \frac{(3 - \sqrt{5})R}{4}$ respectively. In the following we how $A'(x)$, the area of cross section that is only 3-covered,
varies between $x = 0$ to $x =  (\frac{2}{3}  - \frac{\sqrt{10}}{6})$ and between $x = \frac{(3 - \sqrt{5})R}{4}$ to $x = R/2$.
We omit the details of the expression when $x$ is between $(\frac{2}{3}  - \frac{\sqrt{10}}{6})$ and $x = \frac{(3 - \sqrt{5})R}{4}$.
(present in an extended version of the paper)
 
\begin{lemma}
$A'(x) = 8\left[\frac{R^2 - z^2}{2} - \{ (\frac{s^2 \cdot \theta}{2} - \frac{z(z - R/2)}{2}) \} \right] + 8 \left[(\frac{s^2 \cdot \theta'}{2} - \frac{\sqrt{3}R^2}{4}) \} \right] $ for $x = 0$ to $x = (\frac{2}{3}  - \frac{\sqrt{10}}{6}) R$ where $z =  \frac{R + \sqrt{5R^2 - 8x^2 + 8Rx}}{4} $, $\theta = \sin^{-1}(\frac{z}{r})$, $\theta' = \sin^{-1}(\frac{\sqrt{3}}{2})$, $s =\sqrt{Rx - x^2 + \frac{3R^2}{4}}$
\end{lemma}
\begin{proof}
\emph{(Sketch)} Consider the cross-section as depicted in the Figure \ref{z5}. The area that is only 3-covered is the union of 
regions $AEXF, GYHB, TICJ, LSKD$. A simple but tideous calculation gives us the result.   
\end{proof}

\begin{lemma}
$A'(x)$ is equal to $R^2 + 4(2Rx -x^2) \cos^{-1} \left(\frac{R}{2\sqrt{2Rx-x^2}} \right) - \left(2R\sqrt{2Rx - x^2 -\frac{R^2}{4}} \right) 
-\pi(2Rx - x^2)$ for $x = (\frac{3 - \sqrt{5}}{4})R$ to $x= \frac{R}{2}$.
\end{lemma}
\begin{proof}
\emph{(Sketch)} Consider the cross-section as depicted in the Figure \ref{z5}. The area that is only 3-covered is the union of regions $AST,BUV,CWX, ZYD$. 
Again a simple calculation gives us the result.   
\end{proof}

\begin{figure}
\centering
\includegraphics[scale = 0.5]{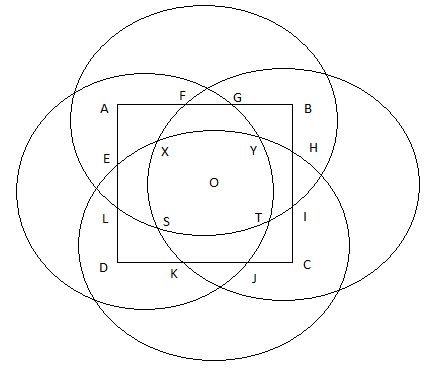}
\caption{For $0 \leq x \leq (\frac{2}{3} - \frac{\sqrt{10}}{6})R$, $A'(x)$ is the union of regions $AEXF, GYHB, TICJ, LSKD$}
\label{z5}
\end{figure}

\begin{figure}
\centering
\includegraphics[scale = 0.5]{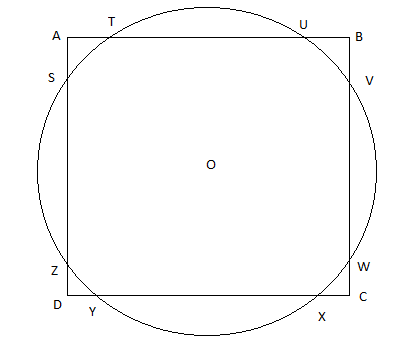}
\caption{For $(\frac{2}{3} - \frac{\sqrt{10}}{6})R \leq x \leq  \frac{(3 - \sqrt{5})R}{4}$, $A'(x)$ is the union of regions $AST,BUV,CWX, ZYD$.}
\label{z6}
\end{figure}

\begin{theorem}
The volume $V'$ inside a cube that is 4-covered is   
\begin{eqnarray*}
V' = 2 \left[ \int^{l_1'}_{0} A(x)\,dx +  \int^{l_2'}_{l_1'} A(x)\,dx +  \int^{R/2}_{l_2'} A(x)\,dx \right]
\end{eqnarray*}
where $l_1' = (\frac{2}{3} - \frac{\sqrt{10}}{6})R$, $l_2' = \frac{(3 - \sqrt{5})R}{4}$
\end{theorem}

We again use numerical methods to estimate this integral, which is approximately $0.048 R^3$.
Thus the volume of region that is 4-covered is approximately $0.952 R^3$.

\section{Deployment Strategy in 3D}

Consider a tessellation of a (poly)cubical 3D Field of Interest (FoI) using cubes
of side length $R$. We place the sensors each of sensing radius
$r = R$ on the centers of the faces of every cube.
Clearly, every sensor will be shared by 2 cubes in the tiling (except the boundary ones).

\begin{lemma}
According to the aforementioned packing every point in the FoI is 3-covered, approximately 
95.2 \% of FoI is 4-covered and about 68.5 \% of the FoI is 6-covered. 
\end{lemma}
\begin{proof}
By construction, every point inside a Sixsoid
formed will be covered by 6 sensors. We just have to
prove that every other point will be 4-covered. According
to the notations used in the construction of Sixsoid
in Section 3, which we borrow here, this is equivalent to
proving that every point in the cross-section (which is
a sqaure of side length R) obtained by slicing the cube
with the plane $P_x$ is 3-covered for all $0 \leq x \leq R$. Again
by symmetry of $S(R)$, we only have to prove this for
all values between $0 \leq x \leq R/2$. As we have noticed in
Section 3 that due to the nature of Sixsoid the topology
of the cross section changes at two values of $x$, namely
$x = l_1$ and $x = l_2$. Consider the arrangement of circles
formed on the square of side length $R$, it can be verified that for
all the ranges when the topology of the cross-section remains fixed
(i.e. from 0 to $l_1$, $l_1$ to $l_2$ and $l_2$ to $R/2$) the entire cross
section is 3-covered. Also from the previous section the claims regarding 
4 and 6 coverage follows.
\end{proof}

Recall that according to our deployment strategy every sensor is
shared by 2 cubes. Thus given a 3D cubical FoI of volume $V_{FoI}$ the number of sensors of radius $r$ needed according
to our deployment strategy is $\frac{3V_{FoI}}{r^3}$ where $r$ is the radius of the sensing 
spheres. 

\begin{table}
\small
 \caption{Volume comparison between Sixsoid and Reuleaux Tetrahedron}
\label{t2}
\begin{center}
    \begin{tabular}{ | l | l | l |}
    \hline
    Volume   & Reuleaux Tetrahedron & Sixsoid  \\ \hline
    $r = 20$  &2787 &  5482   \\ \hline          
    $r = 25$  &5443 & 10708  \\ \hline
    $r = 30$  & 9406 & 18503  \\ \hline
    $r = 35$  &14936 & 29382  \\ \hline
    $r = 40$  &22296 & 43860  \\ \hline
    $r = 45$  &31745 & 62449  \\ \hline
    $r = 50$  &43546 & 85663   \\ \hline
    \end{tabular}
\end{center}
\end{table}

\begin{table}
\small
 \caption{Sensor spatial density comparison between Sixsoid and Reuleaux Tetrahedron for $r = 25$}
\label{t2}
\begin{center}
    \begin{tabular}{ | l | l | l |}
    \hline
    Spatial Density ($\times 10^{-4}$)  & Reuleaux Tetrahedron & Sixsoid  \\ \hline
    $k = 4$  &7.349 &  3.7356   \\ \hline          
    $k = 5$  &9.187 & 4.6694  \\ \hline
    $k = 6$  &11.023 & 5.6033  \\ \hline
    $k = 7$  &12.86 & 6.5372  \\ \hline
    $k = 8$  &14.697 & 7.4711  \\ \hline
    \end{tabular}
\end{center}
\end{table}

\section{Comparison of Two Models}
In this section we present the comparison of our model with the Releaux Tetrahedron Model.
We compare the volume of Sixsoid with the volume of Reuleaux Tetrahedron with varying sensing radius. This comparison is tabulated in Table \ref{t2}. Next, we evaluate the sensor spatial density for $k$-coverage in the Sixsoid model and compare it with the Reuleaux Tetrahedron model. Recall that the minimum sensor spatial density per unit volume needed for full $k$-coverage of a 3D field is defined as a function of sensing radius $r$ as $\lambda = \frac{k}{vol(\mathcal{S}(R))}$ where $R = r$. We take the similar parameter for Reuleaux Tetrahedron from \cite{A,AD1}, which is $\frac{k}{0.422r_0^3}$ where $r_0 = \frac{r}{1.066}$. Table \ref{t3} shows the comparison of the minimum sensor spatial densities and we find that Sixsoid requires much less sensors per unit volume to guarantee $k$-coverage.

\section{Conclusion and Future Work}
In this paper, we discussed a new geometric model for addressing the 
problem of $k$-coverage in 3D homogeneous Wireless Sensor Networks
which we have named as the Sixsoid. We show
in a number of ways how this model outperforms the existing model of Reuleaux 
Tetrahedron \cite{A,AD1}, namely the volume of the convex body involved,
which in turn implies improved spatial density and a better (in terms of coverage)
and more pragmatic deployment strategy for sensors in a given 3D Field of Interest. 
From the point of view of geometry, our construction of Sixsoid and its volume computation 
might be of independent interest. We suspect our construction might have interesting
consequences for non-homogenoeus networks as well which we leave as a direction of future research.

\bibliographystyle{abbrv}

\end{document}